\documentclass[12pt]{article}

\usepackage[english]{babel}
\usepackage[T1]{fontenc}
\usepackage[utf8]{inputenc}
\usepackage{amssymb} 
\usepackage{amsthm} 
\usepackage{mathtools} 
\usepackage{enumerate}

\usepackage[style = authoryear, dashed=false, giveninits=true, maxcitenames=2, uniquename=false, uniquelist=false, maxbibnames=6]{biblatex} 
\usepackage{csquotes} 
\DeclareFieldFormat[article]{title}{#1} 

\renewbibmacro{in:}{}
\DeclareNameAlias{sortname}{last-first} 

\title{Working paper\vspace{1cm}\\Analytic solutions in a continuous-time financial market model}
\author{
Zsolt Bihary\footnote{associate professor at the Department of Finance, Corvinus University of Budapest} \and
Attila A. Víg\footnote{PhD student at the Department of Finance, Corvinus University of Budapest.
\newline email: \texttt{attila.vig@uni-corvinus.hu}
\newline This research was supported by the Higher Education Institutional Excellence Program of the Ministry of Human Capacities in the framework of the `Financial and Public Services' research project (1783-3/2018/FEKUTSTRAT) at Corvinus University of Budapest.
}}

\theoremstyle{definition}
\newtheorem{definition}{Definition}

\theoremstyle{plain}
\newtheorem{proposition}{Proposition}

\begin{document}

\maketitle

\begin{abstract}
    We propose a heterogeneous agent market model (HAM) in continuous time.
    The market is populated by fundamental traders and chartists, who both use simple linear trading rules.
    Most of the related literature explores stability, price dynamics and profitability either within deterministic models or by simulation.
    Our novel formulation lends itself to analytic treatment even in the stochastic case.
    We prove conditions for the (stochastic) stability of the price process, and also for the price to mean-revert to the fundamental value.
    Assuming stability, we derive analytic formulae on how the population ratios influence price dynamics and the profitability of the strategies.
    Our results suggest that whichever trader type is more present in the market will achieve higher returns.
    
    Journal of Economic Literature (JEL) code: G11, G17.
\end{abstract}

\section{Introduction}


Heterogeneous agent models have enjoyed great success in explaining financial market phenomena such as fat tails in return distribution, long-range dependence in volatility, and time series momentum and reversal.
\textcite{brock1998heterogeneous} were among the first to have shown how boundedly rational agents with heterogeneous beliefs can cause market instability when placed into an evolutionary context.
Their discrete time model gave rise to a wide range of models with interacting agents using different expectation schemes (see for example \textcite{hommes2005robust}, \textcite{chiarella2006asset} and \textcite{chiarella2013evolutionary}).

Building on the discrete time framework, a number of continuous-time models have been proposed recently.
Continuous-time models are advantageous when certain delayed variables take part of the system, and they offer more hope for analytic results.
\textcite{he2009market} examine the effect of different lag lengths of moving average rules technical analysts or ``chartists'' use as forecast.
\textcite{he2010dynamics} find that an increase in memory length by the chartists not only can destabilize the market price, but can also stabilize an otherwise unstable market price, leading to stability switching as the memory length increases.
\textcite{he2012heterogeneous} show that with the presence of adaptive heterogeneous agents, a rational route to market instability arises, and the adaptive switching behavior of agents can increase market price fluctuations.
\textcite{he2015profitability} find that the performance of momentum strategy is determined by both time horizon and the market dominance of momentum traders.

A recurring theme in the papers cited is that market stability is examined on the ``deterministic skeleton'' of the otherwise stochastic model, while profitability of the strategies is evaluated by numerical simulations.
In this paper we handle both areas of interest analytically in a stochastic setting.
We propose a simple linear stochastic model with heterogeneous agents in continuous time.
The interacting agents are the usual fundamental traders and chartists.
Market price is arrived at by a market maker scenario.
We define market stability as the asymptotic stationarity of certain well-suited processes, and prove conditions for the market to be stable.
Given stability, we show how the strategic parameters and population ratios of the traders affect the asymptotic stochastic behavior of the market.
We also give analytic formulae for the appropriately defined profitability of the two strategies -- still in a stochastic setting.
Our preliminary results suggest that whichever trader is more present on the market, will achieve higher returns in the long run.

\section{Model}

In this section, we establish an asset pricing model with heterogeneous agents in continuous time.
The approach we take is similar to that of
\textcite{he2010dynamics} and \textcite{he2012heterogeneous}, among others, and is standard in the heterogeneous agent modeling literature.
We incorporate investor heterogeneity and bounded rationality, and examine their effects to the market price behavior.

The model consists of four agents: \emph{fundamental traders}, who take positions according to the fundamental value; \emph{chartists}, who take positions according to time series momentum; \emph{noise traders}, who provide an additional stochastic component in the model; and a \emph{market maker}, who gathers the positions of the traders and adjusts the price accordingly.
The fundamental traders and the chartists are interacting agents, while the noise traders and the market maker are modeled as latent agents.

The financial market we consider consists of a single risky asset (that we think of as the stock market index), and a single risk-free asset (that we think of as the money market account).
The interacting agents build a portfolio of these two assets.
We focus on the market price dynamics itself, and thus we keep the agents' utility maximization problem in the background, and model their portfolio choice directly.

\subsection{Fundamental traders}

Let \(\left(\Omega, \mathcal{A}, P\right)\) be a probability space, let \(W:[0,\infty)\times\Omega\mapsto\mathbb{R}^{2\times 1}\) be a two dimensional Wiener process with independent components representing two different (fundamental and noise trader) risk factors of our model, and let \(\left(\mathcal{F}_t\right)_{t\geq 0}\) denote the natural filtration of \(W\).
All the processes described in the paper are adapted to this filtration.

Let \(S_t\) and \(F_t\) denote the (cum dividend) market price and fundamental value of the risky asset at time \(t\geq 0\) respectively, with \(s_t=\log{S_t}\) and \(f_t=\log{F_t}\) their logarithm.
The fundamental value follows a standard geometric Brownian (GBM) process with drift \(\mu>0\) and volatility \(\sigma_f\in\mathbb{R}^{1\times 2}\), and thus its logarithm follows:
\begin{equation}
df_t=\left(\mu-\frac{\sigma_f\sigma_f^\top}{2}\right)\,dt+\sigma_f\,dW_t \label{eq:fundamental}
\end{equation}
where \(x^\top\) denotes the transpose of \(x\).
We denote by \(u_t\) the price dislocation: the log-difference of the market price and the fundamental value: \(u_t=s_t-f_t\).

We assume that fundamental traders are able to estimate the fundamental value and believe that the market price is mean-reverting to this fundamental value.
Accordingly, the greater (lower) the price dislocation \(u_t\) is, the lower (higher) proportion of their wealth they invest in the risky asset.
Thus, the proportion of their wealth they invest in the risky asset -- portfolio choice for short -- \(Z^f(u_t)\) depends on price dislocation \(u_t\) with \(\frac{\partial Z^f}{\partial u_t}<0\), and we assume a simple affine function similar to \textcite{he2015profitability}:
\begin{equation}
Z^f_t=Z-\alpha_fu_t, \label{eq:portfolio_fundamental}
\end{equation}
where \(\alpha_f\geq 0\) is the fundamental traders' sensitivity parameter for price dislocation and \(Z>0\) is the natural (positive) mean they invest in the risky asset. We assume a positive mean because we assume that investors (both fundamental traders and chartists) are aware of the fundamental value's positive drift.
We note here that fundamental traders play an inherently mean-reverting role in this market: the stronger they are, the more closely the market price follows the fundamental value.

\subsection{Chartists}

Chartists do not know the fundamental value.
Instead, they believe that they can estimate future price movements from past performance of the risky asset.
We assume chartists to be trend followers (momentum traders): they extrapolate recent price movements, and believe that this recent trend will continue.
We define momentum slightly differently than \textcite{chiarella2006asset} and \textcite{he2018asset}.
First, we define \(ds'_t\) as the change of log market price \(s_t\) compared to the average change of log fundamental value \(f_t\):
\[
ds'_t=ds_t-\left(\mu-\frac{\sigma_f\sigma_f^\top}{2}\right)dt
\]
We define momentum \(m_t\) through \(ds'_t\):
\begin{equation}
m_t=\int_{t-\tau}^te^{-k(t-v)}\,ds'_v, \label{eq:momentum_definition}
\end{equation}
with look-back horizon \(\tau>0\) and decay rate \(k>0\).
In words, momentum is defined as the exponentially weighted mean of past log price movements (compared to the average log fundamental value movement) with weights decreasing the further we go back in time.
Using Leibniz's integral rule, \eqref{eq:momentum_definition} can be expressed as a stochastic delay differential equation with time delay \(\tau\):
\begin{equation}
dm_t=ds'_t-e^{-k\tau}\,ds'_{t-\tau}-km_t\,dt. \label{eq:momentum_delay_difference}
\end{equation}
We will devote particular attention to the case \(\tau\to\infty\), in which case \(e^{-k\tau}\to 0\), and \eqref{eq:momentum_delay_difference} becomes a standard stochastic differential equation:
\begin{equation}
dm_t=-km_t\,dt+ds'_t=-km_t\,dt-\left(\mu-\frac{\sigma_f\sigma_f^\top}{2}\right)\,dt+ds_t \label{eq:momentum_difference}
\end{equation}
The \(\tau\to\infty\) case is interesting because \eqref{eq:momentum_difference} is much easier to handle mathematically than \eqref{eq:momentum_delay_difference}, while it is important to note that \(k\) fulfills a similar role to \(\tau\) itself: the relevant history taken into account has a magnitude of \(\frac{1}{k}\).
In this case we can also interpret momentum as the \(\frac{1}{k}\) year log-return.

Chartists believe that recent price movements can be extrapolated: the greater (lower) momentum is, the greater (lower) proportion of their wealth they invest in the risky asset.
Their portfolio choice \(Z^c\left(m_t\right)\) depends on momentum with \(\frac{\partial Z^c}{\partial m_t}>0\) and we assume a simple affine function again:
\begin{equation}
Z^c_t=Z+\alpha_cm_t, \label{eq:portfolio_chartist}
\end{equation}
where \(Z>0\) is the same as in \eqref{eq:portfolio_fundamental}, and \(\alpha_c\geq 0\) is the chartists' sensitivity parameter to momentum.
We note here that chartists play an inherently explosive (with respect to the price dynamics) role in the heterogeneous agent modeling literature.
This role is often alleviated with the use of bounded demand functions.
We assume a simple affine form in \eqref{eq:portfolio_chartist} as opposed to the increasing and bounded S-shaped graph used by \textcite{chiarella2006asset}, because our goal is a linear stochastic model that can be solved analytically.

Our linear model is an approximation at the origin of the S-shaped graph used by the literature.
By this assumption we lose some interesting dynamical aspects in the deterministic case (\(\sigma_f=\sigma_n=0\)), but these interesting aspects are mostly washed away by stochasticity, which we never turn off in this paper.

\subsection{Market maker, noise traders}

Market price is adjusted by the market maker in a similar fashion to \textcite{beja1980dynamic} and \textcite{chiarella2006asset}.
The market maker aggregates the excess demand of fundamental traders and chartists, and adjusts the market price in order to reduce this excess demand while also accounting for the average increase in the fundamental value.
Assume \(p_f\in[0,1]\) and \(p_c\in[0,1]\) with \(p_f+p_c=1\) the ratio of fundamental traders and chartists active on the market, respectively.
The aggregate demand of fundamental traders and chartists is given by
\begin{equation}
\overline{Z}_t=p_fZ^f_t+p_cZ^c_t, \label{eq:aggregate_demand}
\end{equation}
and thus the price setting rule of the market maker is
\begin{equation}
ds_t=\left(\mu-\frac{\sigma_f\sigma_f^\top}{2}\right)\,dt+\beta\left(\overline{Z}_t-Z\right)\,dt, \label{eq:price_setting_rule}
\end{equation}
where \(\beta\geq 0\) is the parameter controlling the speed of price adjustments by the market maker.
We can also think about \(\beta\) as a parameter that determines the strength of supply-demand pressures on the market price. 

We also assume the presence of noise traders or liquidity traders.
We model their perfectly random excess demand through their price impact directly, and complete \eqref{eq:price_setting_rule} with an additive stochastic component to arrive at the final (log) market price dynamics:
\begin{equation}
ds_t=\left(\mu-\frac{\sigma_f\sigma_f^\top}{2}\right)\,dt+\beta\left(\overline{Z}_t-Z\right)\,dt+\sigma_n\,dW_t, \label{eq:price_dynamics}
\end{equation}
where \(\sigma_n\in\mathbb{R}^{1\times 2}\) and \(W_t\) is the same two dimensional Wiener process as in \eqref{eq:fundamental}.

\subsection{Market stability and price dynamics}

Since price dislocation \(u_t=s_t-f_t\) is additive, its change is simply \(du_t=ds_t-df_t\), so we can combine \eqref{eq:fundamental}, \eqref{eq:portfolio_fundamental}, \eqref{eq:portfolio_chartist}, \eqref{eq:aggregate_demand} and \eqref{eq:price_dynamics} to arrive at its dynamics:
\begin{equation}
du_t=\beta\left(p_c\alpha_cm_t-p_f\alpha_fu_t\right)\,dt+\left(\sigma_n-\sigma_f\right)\,dW_t \label{eq:price_dislocation_dynamics}
\end{equation}
The reason we prefer working with \(u_t\) instead of \(s_t\) itself is that the dynamics of \(u_t\) and \(f_t\) together contain all stochastic information about \(s_t\).
Since \(u_t\) and \(f_t\) do not have any explicit connection between them, they can be examined separately.
Also, since \(f_t\) is a random walk process with no stationary distribution, \(s_t\) cannot have a stationary distribution either, whereas \(u_t\) might indeed: if we for example fix \(p_c=0\) in \eqref{eq:price_dislocation_dynamics} (which corresponds to a market with no chartists), \(u_t\) becomes a standard Ornstein-Uhlenbeck process with well-known asymptotically stationary normal distribution.

We can also combine \eqref{eq:momentum_difference} and \eqref{eq:price_dislocation_dynamics} and use the fact \(ds_t=du_t+df_t\) to arrive at the dynamics of \(m_t\) with only \(m_t\) and \(u_t\) on the right hand side:
\begin{equation}
dm_t=\left(\beta p_c\alpha_c-k\right)m_t\,dt-\beta p_f\alpha_fu_t\,dt+\sigma_n\,dW_t \label{eq:momentum_dynamics}
\end{equation}

The full dynamics of our financial market can be expressed as a three-dimensional stochastic differential equation system:
\begin{multline}
d
\begin{pmatrix}
f\\u\\m
\end{pmatrix}_t=
\begin{pmatrix}
\mu-\frac{\sigma_f\sigma_f^\top}{2}\\0\\0
\end{pmatrix}\,dt+
\\
+
\begin{pmatrix}
0& 0& 0\\
0& -\beta p_f\alpha_f& \beta p_c\alpha_c\\
0& -\beta p_f\alpha_f& \beta p_c\alpha_c-k
\end{pmatrix}
\begin{pmatrix}
f_t\\u_t\\m_t
\end{pmatrix}\,dt+
\begin{pmatrix}
\sigma_f\\\sigma_n-\sigma_f\\\sigma_n
\end{pmatrix}\,dW_t \label{eq:full_dynamics}
\end{multline}
We can notice a clear disconnection in \eqref{eq:full_dynamics} between \(f_t\) and \(\left(u_t,m_t\right)^\top\).
Indeed, we examine the latter separately: let \(X_t=\left(u_t,m_t\right)^\top\), and
\begin{equation}
    dX_t=-\Theta X_t\,dt+\Sigma\,dW_t,\quad\text{where} \label{eq:u_m_dynamics}
\end{equation}
\[
\Theta=
    \begin{pmatrix}
    \beta p_f\alpha_f& -\beta p_c\alpha_c\\
    \beta p_f\alpha_f& k-\beta p_c\alpha_c
    \end{pmatrix}
    \in\mathbb{R}^{2\times 2}
\quad\text{and}\quad\Sigma=
    \begin{pmatrix}
    \sigma_n-\sigma_f\\
    \sigma_n
    \end{pmatrix}
    \in\mathbb{R}^{2\times 2}.
\]

\(X_t\) is a two dimensional Ornstein-Uhlenbeck process with no constant drift term, driven by a two dimensional Wiener process with independent components.
Let us denote by \(\sigma_u\) the first row of the volatility matrix, i.e. \(\sigma_u=\sigma_n-\sigma_f\), and also let us introduce \(q_i=\beta p_i\alpha_i\), \(i\in\{f,c\}\).

\begin{definition}
We say that the market defined by \eqref{eq:full_dynamics} is \emph{stable} if its component \(X_t=\left(u_t,m_t\right)^\top\) is (asymptotically) stationary.
\end{definition}
Asymptotic stationarity of the system \(X_t\) implies that neither the component \(u_t\) nor \(m_t\) ``blows up'' -- log market price \(s_t\) does not wander too far from the log fundamental value, thus price dislocation and momentum are ergodic.
We interpret this qualitative property of the market as \emph{stability}.

\begin{proposition}
Conditions
\begin{align}
    \beta\left(p_c\alpha_c-p_f\alpha_f\right)&<k \label{eq:condition1} \tag{C1}
    \\
    \beta p_f\alpha_f&>0 \label{eq:condition2} \tag{C2}
\end{align}
are necessary and sufficient for the market defined by \eqref{eq:full_dynamics} to be stable. \label{prop:market_stable}
\end{proposition}

\begin{proof}
\(X_t\) defined by \eqref{eq:u_m_dynamics} has the solution\footnote{Here \(e^{A}\coloneqq\sum_{k=0}^\infty\frac{A^k}{k!}\), and it holds that \(\left(e^{At}\right)'=Ae^{At}\).}
\begin{equation}
X_t=e^{-\Theta t}X_0+\int_0^te^{\Theta(t-v)}\Sigma\,dW_v, \label{eq:x_solution}
\end{equation}
which is asymptotically stationary if and only if -- according to \parencite{shreve2004stochastic} -- the ordinary differential equation system
\[
\Dot{x}=-\Theta x
\]
is strongly stable, which requires that the real part of all the eigenvalues of \(\Theta\) are strictly positive.
This holds if \(p_c\alpha_c-p_f\alpha_f<\frac{k}{\beta}\) and \(\beta p_f\alpha_fk>0\) both hold.
\end{proof}


We interpret Proposition \ref{prop:market_stable} through price dislocation \(u_t\).
Condition \eqref{eq:condition1} says that market stability requires that the \emph{relative strength} of the chartists compared to that of the fundamental traders \(\left(p_c\alpha_c-p_f\alpha_f\right)\) be not too great.
Here we define \emph{relative strength} as the product of the traders' population ratio \(p_i\) and the sensitivity parameter \(\alpha_i\) of their strategy.
All else being equal, the greater the relative strength of fundamental traders (chartists) is, the stronger their mean-reverting (explosive) price impact is felt on the market.
While chartists are \emph{allowed} to be stronger than fundamental traders, they cannot be \emph{too strong} for the market to be stable.

The maximum difference of relative strength is \(\frac{k}{\beta}\).
The greater the market maker's parameter \(\beta\) for the speed of price adjustment is, the less powerful chartists are allowed to be compared to the fundamental traders: If \(\beta\) is high, chartists' extrapolative effect is more pronounced on the market, and they ``blow up'' the market price more easily.

The inverse is true for decay rate \(k\), which controls the relevant time horizon \(\frac{1}{k}\) when calculating momentum.
The higher \(k\) is, the shorter the history of prices is that is \emph{actually} driving momentum.
When \(k\) is very high, momentum is driven mostly by the perfectly random effect of the noise traders, and chartists are not able to latch onto the ever increasing (or decreasing) market price.

Condition \eqref{eq:condition2} says that the relative strength of the fundamentalists, the market maker's parameter and the decay rate all have to be positive.
This condition says that the only way for the log market price to have any connection to the log fundamental value is that the fundamental traders have positive relative strength on the market.
Also, with no price effect \((\beta=0)\), fundamental traders have no way of pulling the market price towards its fundamental value, no matter how strong they may be.

\begin{proposition}
Assume that the market defined by \eqref{eq:full_dynamics} is stable.
Then, \eqref{eq:x_solution} has a multivariate normal limiting distribution with
\begin{align*}
    \lim_{t\to\infty}E\left[X_t\right]&\overset{.}{=}E\left[X_\infty\right]=
    \begin{pmatrix}
    0\\0
    \end{pmatrix},
    \\
    D^2\left[u_\infty\right]\overset{.}{=}E\left[u^2_\infty\right]
    &=\frac{\left(k\sigma_u+q_c\sigma_f\right)\left(k\sigma_u+q_c\sigma_f\right)^\top+q_fk\sigma_u\sigma_u^\top}{2q_fk\left(q_f+k-q_c\right)},
    \\
    COV\left[u_\infty,m_\infty\right]\overset{.}{=}E\left[u_\infty m_\infty\right]
    &=\frac{\left(q_c-k\right)\sigma_f\sigma_f^\top+k\sigma_n\sigma_n^\top}{2k\left(q_f+k-q_c\right)},
    \\
    D^2\left[m_\infty\right]\overset{.}{=}E\left[m^2_\infty\right]
    &=\frac{q_f\sigma_f\sigma_f^\top+k\sigma_n\sigma_n^\top}{2k\left(q_f+k-q_c\right)},\quad\text{where}\quad\sigma_u=\sigma_n-\sigma_f.
\end{align*}
\label{prop:asymp_distr}
\end{proposition}
\begin{proof}
\eqref{eq:x_solution} is a simple stochastic integral, thus it is normally distributed for all \(t\).
The expectation of the stochastic integral component is 0 for all \(t\), thus
\[
\lim_{t\to\infty}E\left[X_t\right]=\lim_{t\to\infty}e^{-\Theta t}X_0=
\begin{pmatrix}
0\\0
\end{pmatrix}
\]
if the eigenvalues of \(\Theta\) are positive, which holds if \(p_c\alpha_c-p_f\alpha_f<\frac{k}{\beta}\) and \(\beta p_f\alpha_f>0\).

Now let \(\rho:[0,\infty)\mapsto \mathbb{R}^{2\times 2}\), \(\rho_t=COV\left[X_t\right]=E\left[X_tX_t^\top\right]-E\left[X_t\right]E\left[X_t^\top\right]\) and let \(t\) be large enough so that \(\rho_t\approx E\left[X_tX_t^\top\right]\).
Then
\begin{align*}
    d\rho_t
    &=dE\left[X_tX_t^\top\right]=E\left[d\left(X_tX_t^\top\right)\right]=
    \\
    &=E\left[\left(dX_t\right)X_t^\top\right]+E\left[X_t\left(dX_t^\top\right)\right]+E\left[dX_tdX_t^\top\right]=
    \\
    &=-\Theta E\left[X_tX_t^\top\right]\,dt-E\left[X_tX_t^\top\right]\Theta^\top\,dt+\Sigma\Sigma^\top\,dt\Longrightarrow
    \\
    \Longrightarrow \Dot{\rho}&=-\Theta\rho-\rho\Theta^\top+\Sigma\Sigma^\top
\end{align*}
The covariance matrix of the asymptotic distribution is got by setting \(\Dot{\rho}=0\), and solving the resulting system of algebraic equations.
\end{proof}
We interpret Proposition \ref{prop:asymp_distr} by examining some extreme cases.
By setting \(q_c=0\) we arrive at the dynamics where chartists have no effect -- this might happen when their population \(p_c\) is zero and/or their sensitivity to momentum \(\alpha_c\) is zero.
In this case price dislocation \(u_t\) is a simple one dimensional Ornstein-Uhlenbeck process with asymptotic variance \(E\left[u^2_\infty\right]=\frac{\sigma_u\sigma_u^\top}{2q_f}\).

The stronger fundamental traders are, the lower the asymptotic variance of \(u_t\) is -- they keep the market price closer to its fundamental value in the long run.
The inverse is true for the chartists: \(q_c\) appears in the denominator with a negative sign and in the numerator with a positive sign, thus the stronger they are, the greater the variance of \(u_t\) is.

Now assume as a benchmark model a market where noise traders are so abundant that neither fundamental traders nor chartists have any effect on market price.
We can implement this in our model by setting \(\beta=0\) in \eqref{eq:price_dynamics}.
The dynamics of momentum \eqref{eq:momentum_dynamics} in this case is again a simple one dimensional Ornstein-Uhlenbeck process with \(E\left[m^*_\infty\right]=0\) and \(E\left[\left(m^*\right)^2_\infty\right]=\frac{\sigma_n\sigma_n^\top}{2k}\).
We will consider this value as the base variance of momentum (or the base variance of the \(\frac{1}{k}\) year log return).
Let us compare this base value to the variance of momentum of the full model:
\begin{equation}
\frac{E\left[m^2_\infty\right]}{E\left[\left(m^*\right)^2_\infty\right]}=\frac{q_f}{q_f+k-q_c}\frac{\sigma_f\sigma_f^\top}{\sigma_n\sigma_n^\top}+\frac{k}{q_f+k-q_c} \label{eq:momentum_variance}
\end{equation}
For our interpretation we will assume that the volatility of fundamental value is lower than that of the market price: \(\sigma_f\sigma_f^\top<\sigma_n\sigma_n^\top\) (see for example \textcite{shiller1981stock} who shows that market price variance is a magnitude greater than the variance of fundamental value derived as the present value of actual future dividends).
Now assume no chartists on the market, i.e. \(q_c=0\).
In this case \eqref{eq:momentum_variance} is the convex combination of a number lower than one and one, thus it is lower than one.
This means that the presence of fundamental traders reduces the variance of \(\frac{1}{k}\) year log returns.
With \(q_c>0\) the ratio above increases -- indeed, it can even rise above one.
This implies that the presence of chartists increases the variance of \(\frac{1}{k}\) year log returns.

\subsection{Profitability of the strategies}

We assume that both the fundamental traders and the chartists build a self-financing portfolio of the stock \(S_t\) and some money-market account \(B_t\).
We denote by \(V^i_t\), \(i\in\{f,c\}\) the value of their portfolio at time \(t\).
Without loss of generality, we assume a zero risk-free rate: \(\frac{dB_t}{B_t}=0\).
Since both traders hold \(Z^i_t\), \(i\in\{f,c\}\) proportion of their wealth in the risky asset, we arrive at the dynamics of their portfolio by using the self-financing assumption:
\begin{equation}
dV^i_t=Z^i_tV^i_t\,\frac{dS_t}{S_t}+\left(1-Z^i_t\right)V^i_t\,\frac{dB_t}{B_t}=Z^i_tV^i_t\,\frac{dS_t}{S_t} \label{eq:self_financing}
\end{equation}

We assume a logarithmic utility function of their wealth for both the fundamental traders and chartists.
We are interested in the differences in utility over a very long time horizon, where the system \((u_t,m_t)^\top\) has already reached stationarity.
Thus we use an appropriately normalized utility function:
\begin{equation}
\Pi_i=\lim_{T\to\infty}E\left[\frac{\log{V^i_T}-\log{V^i_0}}{T}\right], \label{eq:utility}
\end{equation}
which from now on we will interpret as the average logarithmic growth of the value of the agents' portfolio over an infinite time horizon.
Since we know the asymptotic distribution of \(\left(u_t,m_t\right)^\top\), \eqref{eq:utility} provides analytically tractable results.

\begin{proposition}
Assume the market defined by \eqref{eq:full_dynamics} is stable.
Assume the investors' wealth evolves according to \eqref{eq:self_financing}.
Then the average logarithmic growth \eqref{eq:utility} of the traders' wealth is:
\begin{align*}
    \Pi_f&=C+\beta p_f\alpha_f^2E\left[u^2_\infty\right]-\beta p_c\alpha_c\alpha_fE\left[u_\infty m_\infty\right]-\alpha_f^2\frac{\sigma_n\sigma_n^\top}{2}E\left[u^2_\infty\right]
    \\
    \Pi_c&=C+\beta p_c\alpha_c^2E\left[m^2_\infty\right]-\beta p_f\alpha_f\alpha_cE\left[u_\infty m_\infty\right]-\alpha_c^2\frac{\sigma_n\sigma_n^\top}{2}E\left[m^2_\infty\right]
    \\
    \text{where}\quad C&=Z\left(\mu-\frac{\sigma_f\sigma_f^\top}{2}+(1-Z)\frac{\sigma_n\sigma_n^\top}{2}\right)
\end{align*} \label{prop:utilities}
\end{proposition}
\begin{proof}
We show the calculations for \(\Pi_f\), which go similarly for \(\Pi_c\).

First, from \eqref{eq:price_dynamics} we get the dynamics of \(S_t=e^{s_t}\):
\[
\frac{dS_t}{S_t}=\left(\mu-\frac{\sigma_f\sigma_f^\top}{2}+\frac{\sigma_n\sigma_n^\top}{2}+\beta\left(\overline{Z}_t-Z\right)\right)\,dt+\sigma_n\,dW_t
\]
From \eqref{eq:self_financing} we get the dynamics of \(\log{V^f_t}\):
\[
d\left(\log{V^f}\right)_t=Z^f_t\,\frac{dS_t}{S_t}-\frac{\left(Z^f_t\right)^2}{2}\,\frac{d\langle S\rangle_t}{S^2_t}
\]
By using the definitions \eqref{eq:portfolio_fundamental}, \eqref{eq:portfolio_chartist} and \eqref{eq:aggregate_demand}, we get that
\[
d\left(\log{V^f}\right)_t=g\left(u_t,m_t,u_tm_t,u^2_t\right)\,dt+Z^f_t\sigma_n\,dW_t
\]
where \(g:\mathbb{R}^4\mapsto \mathbb{R}\) is an affine function.
From here, by using the law of iterated expectations\footnote{Here we denote the conditional expectation by \(E_t\), i.e. \(E\left[X_T\middle|\mathcal{F}_t\right]\overset{.}{=}E_t\left[X_T\right]\).} and the fact that we know the asymptotic distribution of \(\left(u_t,m_t\right)^\top\) we get that
\begin{multline*}
\Pi_f=\lim_{T\to\infty}E\left[\frac{\log{V_T}-\log{V_0}}{T}\right]=\\
    \begin{aligned}
    &=\lim_{T\to\infty}\frac{1}{T}E\left[\int_0^T\,d\left(\log{V^f}\right)_t\right]
    =\lim_{T\to\infty}\frac{1}{T}E\left[\int_0^TE_t\left[d\left(\log{V^f}\right)_t\right]\right]=
    \\
    &=\lim_{T\to\infty}\frac{1}{T}E\left[\int_0^Tg\left(u_t,m_t,u_tm_t,u^2_t\right)\,dt\right]=
    \\
    &=\lim_{T\to\infty}\frac{1}{T}\int_0^Tg\left(E[u_t],E[m_t],E[u_tm_t],E[u^2_t]\right)\,dt=
    \\
    &=\lim_{T\to\infty}\frac{1}{T}\int_0^Tg\left(E[u_\infty],E[m_\infty],E[u_\infty m_\infty],E[u^2_\infty]\right)\,dt=
    \\
    &=g\left(0,0,E[u_\infty m_\infty],E[u^2_\infty]\right)
\end{aligned}
\end{multline*}
\end{proof}

Proposition \ref{prop:utilities} helps us answer which strategy ``wins'' in the long run:
\begin{proposition}
The difference in long-term logarithmic growth of the two strategies is given by:
\begin{multline}
    \Pi_f-\Pi_c=
    \frac{\sigma_n\sigma_n^\top}{2}\left(\alpha_c^2E\left[m^2_\infty\right]-\alpha_f^2E\left[u^2_\infty\right]\right)+
    \left(p_f-p_c\right)\beta\alpha_f\alpha_cE\left[u_\infty m_\infty\right]+
    \\
    +\beta\left(p_f\alpha_f^2E\left[u^2_\infty\right]-p_c\alpha^2_cE\left[m^2_\infty\right]\right)
    \label{eq:growth_difference}
\end{multline}
\label{prop:growth_difference}
\end{proposition}
\begin{proof}
A direct consequence of Proposition \ref{prop:utilities}.
\end{proof}

We interpret Proposition \ref{prop:growth_difference} through the effects of population ratios.
While we should not forget that population ratios do affect all the moments present in \eqref{eq:growth_difference}, it seems as if \(\Pi_f-\Pi_c\) is increasing in \(p_f\).
While the exact effect of population ratios is more complex and should be investigated further, for now we claim (as a first order approximation) that whichever trader is more populous on the market, wins in the long run.

\section{Conclusion}

In this paper, we propose a simple linear stochastic model of a financial market populated by heterogeneous agents.
The agents modeled explicitly are fundamental traders (who buy the risky asset when it is below its fundamental value, and sell when it is above) and chartists (who follow a momentum indicator).
A market maker gathers the positions of the two trader types and adjusts the market price accordingly.
Noise traders are also present and provide volatility in market price.

We define market stability through the asymptotic stationarity of the indicators followed by the two traders (price dislocation and momentum).
We prove stability conditions for, and calculate the asymptotic behavior of the market.
We also provide analytic formulae for the profitability of the two trading strategies.
Our preliminary results suggest that whichever trader type is more present in the market achieves higher returns in the long run.




\printbibliography 

\end{document}